\documentclass[11pt,a4paper,twoside,groupcitations]{article}

\usepackage[T1]{fontenc}
\usepackage[ansinew]{inputenc}
\usepackage[english]{babel}
\usepackage{amsfonts}
\usepackage{amsmath}
\usepackage{bm}
\usepackage{array}
\usepackage{amsthm}
\usepackage{amssymb}
\usepackage{graphicx}
\usepackage{subfigure}
\usepackage{braket}
\usepackage{verbatim}
\usepackage[table]{xcolor}
\usepackage{caption}
\usepackage{cite}
\usepackage{textcomp}
\usepackage[title]{appendix}

\usepackage{multicol}
\usepackage{tikz}
\usetikzlibrary{positioning,arrows}
\usetikzlibrary{decorations.pathmorphing}
\usetikzlibrary{decorations.markings}
\usetikzlibrary{calc,decorations.markings}
\usetikzlibrary{arrows,shapes}
\usetikzlibrary{matrix,arrows}
\usepackage{pgfplots}
\usepackage{xparse}
\newcommand{\be}{\begin{equation}}
\newcommand{\ee}{\end{equation}}

\definecolor{pinegreen}{rgb}{0.0, 0.47, 0.44}

\theoremstyle{definition}

\newtheorem*{assumptions}{Assumptions}

\newtheorem{theorem}{Theorem}

\newtheorem{prop}{Proposition}
\newtheorem{lemma}{Lemma}

\theoremstyle{remark}
\newtheorem*{remark}{Remark}

\def\theequation{\thesection.\arabic{equation}}

\title{\bf Realisations of type~III stress-energy tensors of the 
Hawking-Ellis classification in scalar-tensor gravity} %
\author{Narayan~Banerjee$^{a}$\thanks{E-mail: narayan@iiserkol.ac.in},
Valerio~Faraoni$^{b}$\thanks{E-mail: vfaraoni@ubishops.ca}, 
Robert~Vanderwee$^{b}$\thanks{E-mail: rvanderwee20@ubishops.ca}, 
$\ $and
Andrea~Giusti$^{c}$\thanks{E-mail: agiusti@phys.ethz.ch}
\\
\\
$^a${\em Department of Physical Sciences,}\\ 
{\em Indian Institute of Science Education and Research Kolkata}
\\
{\em Mohanpur 741~246, West Bengal, India}
\\
\\
$^b${\em Department of Physics \& Astronomy, Bishop's University}
\\
{\em 2600 College Street, J1M~1Z7 Sherbrooke, Qu\'ebec, Canada}
\\
\\
$^c$ {\em Institute for Theoretical Physics, ETH Zurich}
\\
{\em Wolfgang-Pauli-Strasse 27, 8093 Zurich, Switzerland}
}
\begin{document}
\def\theequation{\arabic{section}.\arabic{equation}} 

\maketitle
\begin{abstract}
The ``ugly duckling'' of the Segr\'e-Pleba\'nski-Hawking-Ellis classification of stress-energy tensors is  believed to be either impossible or extremely difficult to realise in Einstein gravity. {\em Effective} stress-energy tensors in alternative gravity offer a wider range of possibilities. We report a class of type~III realisations in ``first-generation'' scalar-tensor and in Horndeski gravity, and their physical interpretation.  The ugly duckling may be a freak of nature of limited importance but it is not physically impossible. 
\end{abstract}
\newpage
\section{Introduction}
\label{sec:1}
\setcounter{equation}{0}

In general relativity, the right hand-side of the Einstein field equations 
\be
G_{ab} := R_{ab}-\frac{1}{2} \,  g_{ab} R =  T_{ab} \label{EFE}
\ee
(where $R_{ab}$ and $R:= {R^a}_a$ are the Ricci tensor and the Ricci 
scalar of the spacetime metric $g_{ab}$, respectively\footnote{We follow 
the notation of Ref.~\cite{Waldbook}, in which the metric signature is 
$({-}{+}{+}{+})$. Units are used in which the speed of light and $8\pi G$ 
(where $G$ is Newton's constant) are unity.}) contains the 
energy-momentum tensor of matter $T_{ab}^\mathrm{(m)} $. The possible 
forms of $T_{ab}^\mathrm{(m)}$ expected on physical grounds have been 
classified by Segre \cite{Segre84}, Pleba\'nski \cite{Plebanski}, and 
Hawking \& Ellis \cite{HawkingEllis}. The least known and least studied 
type in this classification is type~III, in which the stress-energy 
tensor has the form 
\be
    T_{ab}^\mathrm{(m)} = \rho k_a k_b + q_a k_b + q_b k_a
\ee  
where $k^a$ is a null vector field and $q^a$ is spacelike. Because of its 
unknown nature and unfamiliar properties, the type~III stress-energy 
tensor has been named the ``ugly duckling of the Hawking-Ellis 
classification'' \cite{Martin-Moruno:2018coa, Martin-Moruno:2019kzc}. It 
was believed, although without any firm ground, that type~III 
stress-energy tensors are unphysical until Podolsk\'y, \v{S}varc \& 
Maeda
\cite{Podolsky:2018zha}, Martin-Moruno \& Visser 
\cite{Martin-Moruno:2018coa, Martin-Moruno:2019kzc}, and Maeda 
\cite{Maeda:2023oxl} provided examples in 
which a gyraton or exotic Lagrangians realize this particular 
energy-momentum tensor in Einstein gravity.

One could think of obtaining the null vector field $k^a$ of the type~III 
stress-energy tensor as the gradient of a scalar field $\phi$ 
satisfying 
$\nabla_c \phi \nabla^c \phi=0$. It is straightforward to see that only a 
null dust can be obtained in this way in general relativity. However, one 
can turn to alternative theories of gravity with a built-in scalar.

In alternative gravity, the field equations are often  rewritten by moving geometric terms, or terms built out of the extra gravitational degrees of freedom and their derivatives, to the right-hand side of the field equations to make them look like effective Einstein equations~(\ref{EFE}), thus building  an {\em effective} stress-energy tensor $T_{ab}$, which is formally treated as a mass-energy source of curvature, although it does not describe real matter. These effective stress-energy tensors may provide incarnations of the type~III energy-momentum tensor that are difficult to realize explicitly in Einstein gravity. 

The prototype of the alternative theory of gravity is scalar-tensor 
gravity \cite{Bergmann:1968ve, Brans:1961sx, Nordtvedt:1968qs, 
Wagoner:1970vr}, in which a propagating gravitational scalar field $\phi$ 
is added to the usual massless spin-two modes of Einstein gravity 
contained in the metric tensor. Scalar-tensor gravity is the subject of a 
vast literature and was generalized long ago by Horndeski 
\cite{Horndeski}. His theory went largely unnoticed for many years and was 
then rediscovered in the quest for the most general scalar-tensor theory 
of gravity with second order equations of motion. Although this feature 
was eventually found to be a property of the more general Degenerate 
Higher-Order Scalar-Tensor (DHOST) theories 
(\cite{DHOST1, DHOST2, DHOST3, 
DHOST4, DHOST5, DHOST6, DHOST7}, see \cite{DHOSTreview1, DHOSTreview2} for 
reviews), Horndeski gravity has become the subject of a vast literature 
spanning the last decade ({\em e.g.}, \cite{H1,H2, H3, GLPV1, GLPV2, 
Creminellietal18, Langloisetal18, Langlois:2018dxi, Kobayashi:2011nu}).

Here we analyze ``old'' scalar-tensor gravity first, and then Horndeski 
gravity, and we report a class of possible implementations of the ugly 
duckling type~III {\em effective} stress-energy tensor in these theories.

Independent motivation for our study comes from a completely different 
direction. Recently, the analogy between the effective stress-energy 
tensor of scalar-tensor and viable Horndeski gravity and a dissipative 
(Eckart) fluid has led to introducing an effective ``temperature of 
gravity'' with respect to general relativity and to equations describing 
the approach of alternative gravity to general relativity (or its 
departures from it) \cite{Faraoni:2018qdr, Faraoni:2021lfc,
Faraoni:2021jri,  Giusti:2021sku, Giardino:2022sdv,Giusti:2022tgq, 
Faraoni:2022gry, Miranda:2022wkz, Miranda:2022uyk, Faraoni:2022doe, 
Faraoni:2023hwu, Faraoni:2022jyd, Faraoni:2022fxo, Giardino:2023ygc}. This 
formalism, dubbed ``first-order thermodynamics of scalar-tensor gravity'' 
is subject to the fundamental limitation that the gradient of the 
Brans-Dicke-like scalar field be timelike and future-oriented. We would 
like to extend this formalism to situations in which $\nabla_a \phi$ is 
lightlike instead. As we will see, the goal of 
introducing an effective temperature of gravity turns out to be 
impossible 
but in the process we discover new implementations of the type~III 
stress-energy tensor, only for effective instead of real fluids.

Let us proceed with the definition of this effective energy-momentum 
tensor. The (Jordan frame) action of ``first-generation'' scalar-tensor 
gravity is \cite{Brans:1961sx,Bergmann:1968ve, Nordtvedt:1968qs, 
Wagoner:1970vr}

\be 
\label{eq:BDaction}
S= \int d^4 x \, \sqrt{-g} \left[ \phi R -\frac{\omega(\phi)}{\phi} \, 
\nabla^c \phi \nabla_c \phi -V(\phi) \right] + S^\mathrm{(m)} \,,
\ee 
where $R$ is the Ricci scalar of the spacetime metric $g_{ab}$ with 
determinant $g$ and covariant derivative $\nabla_a$, $\phi$ is the Brans-Dicke-like gravitational scalar with potential $V(\phi)$, $\omega(\phi)$ is the ``Brans-Dicke coupling'', and $S^\mathrm{(m)}$ is the matter action. The corresponding field equations are 
\cite{Brans:1961sx, Bergmann:1968ve, Nordtvedt:1968qs, Wagoner:1970vr}
\be
    G_{ab}= \frac{ T_{ab}^\mathrm{(m)}}{\phi} +\frac{\omega}{\phi^2} 
    \left( \nabla_a \phi \nabla_b \phi -\frac{1}{2} \, g_{ab} \nabla^c \phi 
    \nabla_c\phi \right)
    +\frac{1}{\phi} \left( \nabla_a \nabla_b \phi -g_{ab} \Box \phi \right) 
    -\frac{V}{2\phi} \, g_{ab} \,, \label{fe1}
    \ee
    
    \be
    \left( 2\omega+3 \right) \Box \phi =8\pi T^\mathrm{(m)} +\phi \, 
    V' -2V -  \omega' \, \nabla^c \phi \nabla_c\phi 
    \,, \label{wavelike}
\ee
where $T_{ab}^\mathrm{(m)}$ is the matter stress-energy tensor, $\Box 
\equiv \nabla^c \nabla_c$, and a prime denotes differentiation with 
respect to $\phi$.  In the next section we restrict ourselves to 
solutions 
of these field equations with the property that the scalar field gradient 
is null, $\nabla^c \phi \nabla_c\phi=0$. Examples of such solutions are 
reported in 
Refs.~\cite{Tupper74, Ray:1977nr, Gurses:1978zs, Bressange:1997ey, 
Gurses:2016vwm, Racsko:2018xcw, Siddhant:2020gkn, Gurses:2021ogc}.

\section{Null scalar field gradient in first-generation scalar-tensor gravity}
\label{sec:2}
\setcounter{equation}{0}
Let $\phi$ be the Brans-Dicke scalar field and $k_a := \nabla_a 
\phi$. We shall now present a few preliminary results.
\begin{lemma} \label{lemma-1}
Consider the Brans-Dicke theory \eqref{eq:BDaction} with $\omega \neq 
-3/2$, $V=0$, and {\em in vacuo}
($T_{ab}^\mathrm{(m)}=0$). If $k^a$ 
is a null vector field ({\em i.e.}, $k_a k^a=0$), then $\Box \phi = 0$.
\end{lemma}

\begin{proof}
The field equation for the Brans-Dicke scalar field (\ref{wavelike}) for 
$\omega \neq -3/2$, $V=0$, and  {\em in vacuo}  reduces to 
$$
\left( 2\omega+3 \right) \Box \phi = -  \omega' \, \nabla^c \phi 
\nabla_c\phi \, .
$$
If now the scalar field gradient is null, then $\nabla^c \phi 
\nabla_c\phi = k_c k^c = 0$, that implies $\left( 2\omega+3 \right) \Box 
\phi = 0$, or simply $\Box \phi = 0$.
\end{proof}

\begin{prop} \label{prop-2}
Consider the Brans-Dicke theory \eqref{eq:BDaction} with $\omega \neq 
-3/2$, $V=0$, and {\em in vacuo}
($T_{ab}^\mathrm{(m)}=0$). 
If $k^a$ is a null vector field, then $k^a$ is a geodesic vector field and the corresponding geodesic is affinely parametrised.
\end{prop}

\begin{proof}
Recalling the definition $k_a := \nabla_a \phi$ one has that
$$
k^a \nabla_a k_b = k^a \nabla_a \nabla_b \phi = k^a \nabla_b \nabla_a \phi = k^a \nabla_b k_a \, .
$$
Differentiating $k^a k_a = 0$ one gets $k^a \nabla_b k_a = 0$ that, together with the previous equation, allows one to conclude that $k^a \nabla_a k_b = k^a \nabla_b k_a = 0$. In other words, $k^a$ is a null geodesic vector field and the null geodesic curve to which it is tangent is affinely parametrised.
\end{proof}

\begin{remark}: This property holds for both minimally and 
non-minimally coupled scalar fields $\phi$ irrespective of whether $\phi$ is a test field or a 
gravitating one. \end{remark}

Consider now the congruence of null geodesics with tangent field $k_a 
=\nabla_a \phi$. Following standard procedure (see, {\em e.g.}, 
\cite{Poisson:2009pwt}) and adapting to the special situation in which the 
null field $k^a$ is a gradient, we define another (non-unique) null vector 
field $n^a$ normalized so that
\be \label{defn}
    k_a n^a =-1 \, .
\ee
The covariant differentiation of this equation leads to the useful relation
\be
    n_b \nabla_a k^b = - k_b \nabla_a n^b \, .
\ee
The 2-metric transverse  to both $k^a$ and $n^b$ is defined as
\be
    h_{ab} :=  g_{ab} + k_a n_b + k_b n_a \, ; \label{hab}
\ee
$h_{ab}$ satisfies
\be
    h_{ab} k^a = h_{ab} k^b = h_{ab} n^a = h_{ab} n^b = 0 \,,
    \ee
is 2-dimensional ($ {h^a}_a = 2 $),  and
\be
{h^a}_c \, {h^c}_b = {h^a}_b \, .
\ee
Defining now 
\be
    B_{ab} := \nabla_b k_a \,,
\ee
it follows immediately  that $B_{ab}$ is symmetric,
\be
    B_{ab} = \nabla_b \nabla_a \phi =
    \nabla_a \nabla_b \phi = B_{ba} = B_{(ab)} 
\ee
which, in conjunction with the fact that $k^a$ is geodesic, implies that 
$B_{ab}$ is transverse to $k^a$, {\em i.e.},
\be
    B_{ab}k^b = 
    B_{ab}k^a = 0 \, . \label{eq:transversality1}
\ee

Next, we project $B_{ab}$ onto the 2-space orthogonal to both $k^a$ and 
$n^a$, obtaining
\be 
    \tilde{B}_{ab} := {h_a}^c {h_b}^d B_{cd} \, .
\ee
We can write this quantity explicitly in terms of $k^a$ and $n^a$ using Eq.~(\ref{hab}), which gives
\be
    \tilde{B}_{ab} = B_{ab} + k_a n^c B_{cb} + k_b n^c B_{ac} + k_a k_b \left( n^c n^d B_{cd} \right) \, . \label{tildeB}
\ee
As with any rank two tensor, $\tilde{B}_{ab}$ can be decomposed into its symmetric and anti-symmetric parts,
\be
    \tilde{B}_{ab} = \Theta_{ab} + \tilde{\omega}_{ab} \, ,
\ee
\be
    \Theta_{ab} := \tilde{B}_{(ab)} = \frac{\Theta}{2} \, h_{ab} +  
\tilde{\sigma}_{ab} \, , \quad \tilde{\omega}_{ab} := \tilde{B}_{[ab]} = 
0 \,, \ee
where $\Theta \equiv \Theta^a_{\phantom{a}a} =\nabla_c k^c$ is the 
expansion scalar while $\tilde{\sigma}_{ab}$ is the shear tensor, which is 
the symmetric, trace-free part of $\tilde{B}_{ab}$, and the vorticity $ 
\tilde{\omega}_{ab}$ vanishes identically because $k_a$ is a gradient 
\cite{Faraoni:2021jri}.

From this point on let us restrict our attention to the following case.
\begin{assumptions}[$\star$]
{\em Vacuum} Brans-Dicke gravity with $V = 0$, $\omega \neq -3/2$, and 
$k_c=\nabla_c \phi$ a null gradient of the Brans-Dicke scalar field.
\end{assumptions}

In light of {\bf Lemma \ref{lemma-1}} we have  
$\Theta = \nabla_c k^c= \Box \phi = 0$, then it holds that
\be 
    \tilde{B}_{ab} = \tilde{\sigma}_{ab} \, .
\ee

This particular scenario of Brans-Dicke gravity admits the possibility 
that $k^a$ is a Killing vector field. In this case, $B_{ab}=0$ and 
$T_{ab}=\frac{\omega}{\phi^2} \, k_a k_b$ describes a pressureless null 
dust \cite{Faraoni:2018fil}. We do not consider this very special 
situation further.

In general,  the scalar field effective stress-energy  tensor for our 
subclass of Brans-Dicke gravity $\bm{(\star)}$ reads
\be
    T_{ab} = \frac{\omega}{\phi^2} \nabla_a \phi \nabla_b \phi + 
\frac{\nabla_a \nabla_b \phi}{\phi} = \frac{\omega}{\phi^2} k_a k_b + 
\frac{B_{ab}}{\phi}
\ee
and, using Eq.~(\ref{tildeB}), one can write this $T_{ab}$ in the form
\be
    T_{ab} = \left[\frac{\omega}{\phi^2} - \left(\frac{n^c n^d B_{cd} 
}{\phi} \right)\right] k_a k_b + \frac{\tilde{B}_{ab}}{\phi} + q_a k_b + 
q_b k_a \, , \label{SEtensor}
\ee
where
\be
\label{def-q}
     q_a := -\frac{n^c B_{ca}}{\phi} =-\frac{n^c \nabla_a k_c}{\phi} = 
\frac{k^c \nabla_a n_c}{\phi}  \, .
\ee
Note that
\be \label{eq:orthoqk}
q_a k^a=0 
\ee
because of the $k$-tranversality of $B_{ab}$. 

\begin{prop} \label{prop-3}
Let $u^a$ be a null vector field and let $v^a$ be a vector field 
orthogonal to $u^a$, {\em i.e.}, $u^a v_a = 0$. Then $v^a$ is either 
parallel to $u^a$ or it is spacelike.
\end{prop}

\begin{proof}
If $v^a$ is parallel to $u^a$, namely $\exists \alpha \in \mathbb{R}$ 
such that $v^a = \alpha \, u^a$, then $u^a v_a = \alpha \, u^a u_a = 0$.

\noindent If $v^a$ is not parallel to $u^a$ we have to show that $v^a v_a 
> 0$. Consider a local inertial frame at each spacetime point $p$ where 
$u^a v_a =0$ such that $u^\mu = (u, 0, 0, u)$, with $u\in \mathbb{R}$ in this chart. In this local inertial frame the components of $v^a$ will read, in general, as $v^\mu = (v^0, v^1, v^2, v^3)$, with $v^0, \ldots, v^3 \in \mathbb{R}$. The orthogonality condition $u^a v_a = 0$ at $p$ in the local frame reads
$$
0 = \eta _{\mu \nu} \, u^\mu v^\nu = - u v^0 + u v^3  \, , 
$$ 
with $\eta_{\mu \nu}$ denoting the Minkowski metric. Therefore in the 
local frame at each spacetime point where $u^a v_a =0$ one has that 
$v^\mu = (v, v^1, v^2, v)$ with $v \in \mathbb{R}$. This implies that 
$v^\mu v_\mu = (v^1)^2 + (v^2)^2 > 0$ at $p$. In other words, at each 
spacetime point $p$ where $u^a v_a =0$ it holds that $v^a v_a > 0$, if 
$v^a$ is not parallel to $u^a$.
\end{proof}

Therefore, Eq.~\eqref{eq:orthoqk} and {\bf Prop. \ref{prop-3}} tells us 
that $q^a$ is either parallel to $k^a$ or $q^a$ must be spacelike.

\subsection{Case 1: $q^a$ parallel to $k^a$}
\label{subs:Case1}
First we observe that:
\begin{lemma}
Given {\bf Assumption} $\bm{(\star)}$, if $q_a := - n^c B_{ca} / \phi$ is parallel to 
$k^b$ then $q^a = -\left( q_c n^c \right)  k^a$.
\end{lemma}

\begin{proof}
Let $q^a$ be parallel to $k^a$, this means that $\exists \alpha 
\in  \mathbb{R}$ such that $q^a = \alpha k^a$. From the definition of 
$n^a$,  Eq.~\eqref{defn}, one has that $n^a q_a = \alpha n^a k_a = - 
\alpha$. Hence $q^a = -\left( q_c n^c \right) k^a$.
\end{proof}

The effective stress-energy tensor~(\ref{SEtensor}) of the scalar field, 
given the subclass of Brans-Dicke gravity $\bm{(\star)}$, becomes
\be \label{218}
    T_{ab} = \left[\frac{\omega}{\phi^2} - (q_c n^c)\right] k_a k_b + 
\frac{\tilde{B}_{ab}}{\phi} 
\ee
with energy density 
\be
    \rho =  \frac{\omega}{\phi^2} - \left( q_c n^c \right) \, ,
\ee
and 
\be
    \frac{\tilde{B}_{ab}}{\phi} = \frac{\tilde{\sigma}_{ab}}{\phi}  \, .
\ee
This allows one to introduce an effective, trace-free, anisotropic stress 
tensor
\be
\tilde{\pi}_{ab} := \frac{\tilde{\sigma}_{ab}}{\phi} = -2 \eta \, 
\tilde{\sigma}_{ab} \, ,
\ee
where $\eta = - 1/ \left( 2\phi \right) $ is an effective  shear 
viscosity 
coefficient. The effective energy-momentum tensor \eqref{218} of $\phi$ 
then takes the form of a null fluid:
\be
T_{ab}= \rho k_a k_b + \tilde{\pi}_{ab} \, . \label{SEnullfluid}
\ee
One can always diagonalize the anisotropic stress tensor  
$\tilde{\pi}_{ab}$ via a rotation of axes since $\tilde{\sigma}_{ab}$ is 
a symmetric tensor in a Riemannian 2-dimensional space. 
There are two spacelike vectors $x^a$ and $y^a$ in this 2-space 
orthogonal to both $k^a$ and $n^a$ (with 
metric $h_{ab}$) such that
\be
    x_c k^c = y_c k^c = x_c n^c = y_c n^c = x_c y^c = 0 \,  
\ee
and
\be
    x_c x^c = y_c y^c = 1
\ee
for which $\tilde{\pi}_{ab}$ is diagonal. In this coordinate system, the 
effective stress-energy  tensor has the null fluid form \cite{Waldbook}
\be
    T_{ab} = \rho k_a k_b + P_1 \left( x_a x_b - y_a y_b \right)  \,.
\ee
In this case the effective stress-energy tensor of $\phi$ describes a 
null fluid with 
anisotropic stresses $ \tilde{\pi}_{ab}$ satisfying the constitutive 
relation of a Newtonian fluid $ \tilde{\pi}_{ab} 
= -2\eta \tilde{\sigma}_{ab}$, with shear viscosity $\eta = -1/\left( 
2\phi \right) $, no heat conduction, and 
vanishing trace ${T^a}_a$. 

We shall now show that the shear can be eliminated, reducing this null 
fluid to a null dust (type~II in the Hawking-Ellis classification 
\cite{HawkingEllis}).

Consider a congruence of null geodesics with tangent $k_a=\nabla_a \phi$. 
The proof of this second statement uses the Raychauduri equation for null 
geodesic congruences \cite{Waldbook,HawkingEllis}
\be
    \frac{d\Theta}{d\lambda} = -\frac{\Theta^2}{2} - \tilde{\sigma}_{ab} 
\tilde{\sigma}^{ab} + \tilde{\omega}_{ab} \tilde{\omega}^{ab} - R_{ab} 
k^a k^b \,, \label{Raychaudhuri}
\ee
where $\lambda$ is an affine parameter along the null geodesics. In our 
case, the null geodesic congruence with tangent field $k^a$ has $\Theta = 
0$ and $ d\Theta/ d\lambda = 0$ everywhere, hence any pair of initially 
parallel geodesics remains parallel. The vorticity $ \tilde{\omega}_{ab}$ 
also vanishes identically. Since we are assuming vacuum, the trace of the 
matter stress-energy tensor $T^{\mathrm{(m)}}_{ab}$ does not contribute to 
the Ricci scalar, and $\nabla_c \phi\nabla^c\phi=0$, $V=0$ and also $ 
\Box\phi=0$ ({\bf Lemma \ref{lemma-1}}). Contracting the vacuum field 
equation~(\ref{fe1}) yields $R=0$ and the effective Einstein 
equation~(\ref{fe1}) contracted twice with $k^a$ gives
\begin{align}
    R_{ab} k^a k^b 
    & = T_{ab} k^a k^b \\ \nonumber
    & = \left(\rho k_a k_b + \frac{\tilde{B}_{ab}}{\phi} + q_a k_b + q_b k_a\right) k^a k^b = 0 
\end{align}
due to the lightlike nature of $k^c$ and the 
transversality~(\ref{eq:transversality1}) of $\tilde{B}_{ab}$ to $k^c$. To 
conclude, the Raychaudhuri equation~(\ref{Raychaudhuri}) yields 
$2 \tilde{\sigma}^2 := \, \tilde{\sigma}_{ab} \tilde{\sigma}^{ab}=0$ 
everywhere, which implies that, since $\tilde{\sigma}^2$ is positive 
definite, all components of the shear $\tilde{\sigma}_{ab}$ vanish 
identically \cite{Ellis:1971pg} and the effective stress-energy tensor 
\eqref{218} of $\phi$ reduces to
\be
    T_{ab} = \rho k_a k_b \,,
\ee
which is of type~II in the Hawking-Ellis classification system.

These results can therefore be summarised in the following theorem.

\begin{theorem}
Given {\bf Assumption} $\bm{(\star)}$, if $q^a$ is parallel 
to $k^a$ then the 
effective stress-energy tensor~(\ref{SEtensor}) of the scalar field 
reduces to the stress-energy tensor of a null dust, which belongs to the type~II family of the Hawking-Ellis classification.
\end{theorem}

\subsection{Case 2: $q^a$ is spacelike}
\label{subsec:2}

Let us consider now the second possibility, in which $q^a$ is spacelike. 
Then, $q^a$ lives in 
the 2-space orthogonal to both $k^a$ and $n^a$ and $q_a = h_{ab} q^b$. In 
fact, since $q_c k^c=0$, the vector field $q_a$ can only have a component 
parallel to $n_a$ and components in the 2-space orthogonal to both $k^a$ 
and $n^a$:
\begin{eqnarray}
q_a &=& h_{ab} q^b + \left( q_c n^c \right) n_a \nonumber\\
&=& q_a + (q_b n^b ) k_a +\left( q_c n^c\right) n_a 
\nonumber\\
&=& q_a + \left( q_c n^c \right) \left( k_a + n_a \right) 
\,,\label{astrominchia}
\end{eqnarray}
where we used the definition of $h_{ab}$. In 
order for Eq.~(\ref{astrominchia}) to be satisfied, either $q_c n^c=0$ 
(and then $q^a$ lies in the 2-space orthogonal to both $k^a $ and $n^a$, 
$q_a= 
h_{ab} q^b$, $q^{\mu}= \left( 0, q^1, q^2, 0 \right)$), or else 
$n_a= -k_a $, which is impossible because $n^a$ is chosen to be 
independent 
of $k^a$ and to satisfy $k_c n^c=-1$. Therefore, $q^a$ is orthogonal to 
both $k^b$ and $n^b$.

We  now have
\be
q_c n^c = -\frac{ B_{cd} n^c n^d}{\phi} = -\frac{ n^c n^d \nabla_c 
k_d}{\phi} =   \frac{ n^c k^d \nabla_c n_d}{\phi} = 0  
\ee
and the effective stress-energy tensor~(\ref{SEtensor})  reduces to 
\begin{eqnarray}
T_{ab} & = &  \frac{\omega}{\phi^2} \, k_a k_b 
    + \frac{ \tilde{B}_{ab}}{\phi}
    + q_a k_b + q_b k_a \,,  \label{SEtype3}
\end{eqnarray}
where we used the fact that $ q_c n^c = -n^c n^d B_{cd}/\phi =0 $. This energy-momentum tensor can then be written in the form
\be
    T_{ab} = \rho k_a k_b + \tilde{\pi}_{ab} + q_a k_b + q_b k_a \, ,
\ee
with $\rho = \omega/\phi^2$. Furthermore, using again the Raychauduri 
equation for null geodesic congruences we can eliminate the contribution 
of the shear $\tilde{\pi}_{ab}$, thus reducing the stress-energy tensor 
for $\phi$ to 
\be \label{eq:typeIIIBD}
T_{ab} = \rho k_a k_b + q_a k_b + q_b k_a \, ,
\ee
which has vanishing trace (since $k_a k^a = k_a q^a=0$) and falls into 
type III of the Hawking-Ellis 
classification system \cite{HawkingEllis} (see also \cite{Segre84, 
Plebanski}). This type is the least known of this 
classification and it is largely unknown which physical systems can be 
described by a type III  tensor. The 
only known examples, as previously mentioned, are the gyraton and the 
exotic Lagrangians discussed by 
Podolsk\'y, \v{S}varc \& Maeda
\cite{Podolsky:2018zha}, Martin-Moruno \& Visser 
\cite{Martin-Moruno:2019kzc, Martin-Moruno:2018coa, 
Martin-Moruno:2017iqw, Martin-Moruno:2018eil, 
Martin-Moruno:2021niw}, and Maeda 
\cite{Maeda:2023oxl}. The stress-energy tensor~(\ref{SEtype3}) is further 
reduced to the type III$_0$ of  
\cite{Martin-Moruno:2018coa, Martin-Moruno:2019kzc}  if the Brans-Dicke 
coupling assumes the special value $\omega = 0$.

In other words, we have shown that:

\begin{theorem}
Given {\bf Assumption} $\bm{(\star)}$, if $q^a$ spacelike 
then the effective stress-energy tensor~(\ref{SEtensor}) of the scalar 
field reduces to the stress-energy tensor \eqref{eq:typeIIIBD}, which 
belongs to the type~III family of the Hawking-Ellis classification. 
Furthermore, if $\omega = 0$ the stress-energy tensor \eqref{eq:typeIIIBD} 
further reduces to the type III$_0$ class.
\end{theorem}

\section{Null scalar field gradient in viable Horndeski gravity}
\label{sec:3}
\setcounter{equation}{0}

Let us now move to Horndeski theories of gravity \cite{Horndeski}, the 
subject of intense research in the past decade ({\em e.g.,} \cite{H1,H2, 
H3, GLPV1, GLPV2, Creminellietal18, Langloisetal18, Langlois:2018dxi, 
Kobayashi:2011nu} and references therein), which are much more general 
than first-generation scalar-tensor theories and contain them as a special 
case. We restrict ourselves to the so-called viable Horndeski 
theories in which the coupling functions are constrained by the 
requirement that gravitational waves propagate at light speed, as shown by 
the GW170817/GRB170817A multi-messenger event detected in gravitational 
waves and in many electromagnetic bands \cite{TheLIGOScientific:2017qsa, 
Monitor:2017mdv}.

The action of viable Horndeski gravity is 
\be
    S = \int d^{4}x \sqrt{-g} \left( \mathcal{L}_2 + \mathcal{L}_3 + 
\mathcal{L}_4 \right) + S^\mathrm{(m)} \, , \label{vhorndeskiaction}
\ee
where $S^\mathrm{(m)}$ is the matter action. The Lagrangian densities 
$\mathcal{L}_i  \, (i = 2,3,4)$ are
\begin{eqnarray}
    \mathcal{L}_2 & = & G_{2}(\phi,X) \, ,\\
    \mathcal{L}_3 & = &-G_{3}(\phi,X)\Box\phi \,,\\
    \mathcal{L}_4 & = &  G_{4}(\phi) R \,,  \label{lagrangedensities}
\end{eqnarray}
where the $G_i  $ are regular functions of the scalar 
field $\phi$ and of  $X := 
-\frac{1}{2} \, \nabla_{c} \phi \nabla^{c} \phi$ (except for $G_4$ that 
depends only on $\phi$ in viable Horndeski theories). The variation of 
the action 
(\ref{vhorndeskiaction}) with respect to the inverse metric $g^{ab}$ 
yields the effective field equations (see {\em e.g.}, 
\cite{Miranda:2022wkz})
\begin{align}
    G&_4 G_{ab} - \nabla_a \nabla_b G_4 + \left(\Box G_4 - \frac{G_2}{2} 
- \frac{1}{2} \nabla_c \phi \nabla^c G_3 \right) g_{ab} \nonumber\\ 
    & + \left( \frac{G_{3X} }{2} \, \Box \phi - \frac{ G_{2X}}{2}  
\right) \nabla_a \phi \nabla_b \phi + \nabla_{(a} \phi \nabla_{b)} G_3 = 
 T^\mathrm{(m)}_{ab} \, ,  \label{gabEOM}
\end{align} 
and variation with respect to $\phi$ gives the equation of motion for the 
scalar field (see {\em e.g.}, \cite{Miranda:2022wkz}):
\begin{align}
    G_{4\phi} & R  + G_{2\phi} + G_{2X} \Box \phi + \nabla_c \phi 
\nabla^c G_{2X} - G_{3X}(\Box \phi)^2 \nonumber \\
    & - (\nabla_c \phi \nabla^c G_{3X}) \Box \phi -    
    G_{3X} \nabla^c \phi \Box \nabla_c \phi   \nonumber\\
    & + G_{3X} R_{ab} \nabla^a \phi \nabla^b \phi - 
\Box G_3 - G_{3\phi} \Box \phi = 0 \,, \label{whichlabeldoIput}
\end{align}
where $T^\mathrm{(m)}_{ab}$ is the matter stress-energy tensor and
\be
    G_{i\phi} := \frac{\partial G_i}{\partial \phi} \, , \quad G_{iX} := \frac{\partial G_i}{\partial X} \quad (i=2,3,4) \, .
\ee
The field equations~(\ref{gabEOM}) can be cast as the effective Einstein equations
\be
    G_{ab} = T_{ab} + \frac{T^\mathrm{(m)}_{ab}}{G_4} \,,
\ee
where $T_{ab} = T^{(2)}_{ab} + T^{(3)}_{ab} + T^{(4)}_{ab}$ is the scalar field effective 
stress-energy tensor capturing all deviations from GR, with  
\begin{eqnarray}
T^{(2)}_{ab}  &=& \frac{1}{2G_4} \left( G_{2X} \nabla_a \phi \nabla_b \phi + G_2 g_{ab} \right) \, , \\
T^{(3)}_{ab} &=& \frac{1}{2G_4} \left( G_{3X} \nabla_c X \nabla^c \phi - 2XG_{3 \phi} \right) g_{ab} \nonumber\\
   & & - \frac{1}{2G_4} \left( 2G_{3 \phi} + G_{3X} \Box \phi \right)\nabla_a \phi \nabla_b \phi - \frac{G_{3X}}{G_4} \nabla_{(a} X \nabla_{b)} \phi \,, \\
T^{(4)}_{ab} &=& \frac{G_{4\phi}}{G_4} \left( \nabla_a \nabla_b \phi - g_{ab} \Box \phi \right) 
 + \frac{G_{4\phi\phi}}{G_4} \left( \nabla_a \phi \nabla_b \phi + 2X g_{ab} \right) \,.
\end{eqnarray}
When the scalar field gradient is null at every spacetime point 
 \cite{Gurses:1978zs, Gurses:2016vwm, Gurses:2021ogc} and 
$\Box\phi=0$, the canonical kinetic term $X$ vanishes with all its 
derivatives, leading to the much simpler total effective energy-momentum 
tensor
\be
T_{ab} = \left(\frac{G_{2X}}{2G_4} - \frac{G_{3\phi}}{G_4} + 
\frac{G_{4\phi\phi}}{G_4}\right) \nabla_a \phi \nabla_b \phi + 
\frac{G_2}{2G_4} g_{ab} + \frac{G_{4\phi}}{G_4} \, \nabla_a \nabla_b \phi 
\,.
\ee
It is shown in Appendix~\ref{sec:AppendixA} that the assumptions 
of this section imply a relation between Horndeski coupling functions
at $X=0, \, \Box\phi=0$. It is also possible to show that the scalar field 
gradient 
$\nabla^a\phi$ is an eigenvector of the Ricci tensor 
(Appendix~\ref{sec:AppendixA}). For lightlike $k^a = \nabla^a \phi$, we  
rewrite this effective $T_{ab}$ as  
\be
T_{ab} = \left(\frac{G_{2X}}{2G_4} - \frac{G_{3\phi}}{G_4} + 
\frac{G_{4\phi\phi}}{G_4}\right) k_a k_b 
+ \frac{G_2}{2G_4} g_{ab} + \frac{G_{4\phi}}{G_4} \nabla_b k_a 
\,.\label{SEintermed}
\ee
Following the logic of the previous section, we define the tensor $B_{ab} 
\equiv  \nabla_b k_a$ and the 2-metric $h_{ab} \equiv g_{ab} + k_a n_b + 
k_b n_a$ orthogonal to both $k^a$ and the  
auxiliary null vector $n^a$, then $\tilde{B}_{ab} = h_a^{\phantom{a}c} 
h_b^{\phantom{b}d} B_{cd}$ is again given explicitly by 
Eq.~(\ref{tildeB}). The effective stress-energy tensor 
obtained from substituting these expressions into (\ref{SEintermed}) 
reads 
\begin{eqnarray}
    T_{ab} & = & \left(\frac{G_{2X}}{2G_4} - \frac{G_{3\phi}}{G_4} + 
     \frac{G_{4\phi\phi}}{G_4} - \frac{ G_{4\phi} n^c n^d B_{cd}}{G_4} 
     \right) k_a k_b \nonumber\\ 
    &&\nonumber\\
    & \, &  + \frac{G_2}{2G_4} \, h_{ab} + \frac{ G_{4\phi}}{G_4} \, 
	\tilde{\sigma}_{ab} \nonumber\\
    &&\nonumber\\
    & \, & + \left(-\frac{G_{4\phi}}{G_4} \, n^c B_{ca} - 
    \frac{G_2}{2G_4} \, n_a \right) k_b \nonumber\\ 
    &&\nonumber\\
    & \, & + \left( -\frac{G_{4\phi}}{G_4} n^d B_{db} - \frac{G_2}{2G_4} 
    \, n_b \right) k_a \, . \label{nullphiTab}
\end{eqnarray}
In this form, it is straightforward to identify the relevant fluid 
quantities:
\begin{eqnarray}
    \rho & := & \frac{  G_{2X}  - 2G_{3\phi} + 2G_{4\phi\phi} - 
2G_{4\phi} n^c n^d B_{cd} }{2G_4}  \, , \label{rho} \\ 
    P & := & \frac{G_2}{2G_4} \,, \\
    \tilde{\pi}_{ab} & := &\frac{G_{4\phi}}{G_4} \,  \tilde{\sigma}_{ab} 
\, ,  \label{Pab} \\
    q_a & := & -\frac{G_{4\phi}}{G_4} \, n^c B_{ca} - \frac{G_2}{2G_4} \, 
n_a \, , \label{qa}
\end{eqnarray}
{\em i.e.}, energy density, isotropic pressure, anisotropic stress 
tensor, and energy current density, respectively. Using the 
identifications~(\ref{rho})--(\ref{qa}) the stress-energy 
tensor~(\ref{nullphiTab}) takes the form
\be \label{H-SEtensor}
    T_{ab} = \rho k_a k_b + P h_{ab} + \tilde{\pi}_{ab} + q_a k_b + q_b k_a \, .
\ee

Note that now there appears the isotropic pressure $P=G_2/(2G_4)$, which 
was absent in first-generation scalar-tensor gravity because, there, $G_2 
\left( \phi, X \right) = \omega(\phi) X/2$ vanishes for $X=0$. 
In viable Horndeski,
\be
    q_a k^a =\frac{G_2}{2G_4}=P \neq 0
\ee
in general and the energy flux density $q^a$ no longer lives in the 
2-space with metric $h_{ab}$ orthogonal to both $k^a$ and $n^a$, but has 
components along the light cone generated by these null vectors since $q_a 
k^a \neq 0$, $ q_a n^a \neq 0$. The trace \be
    {T^a}_a = 4P 
\ee
now does not vanish. What is more, in viable Horndeski theory $q^a$ does 
not have a definite causal character because
\be
    q_a q^a = \frac{ G_{4\phi}}{G_4^{\,2}} \left( G_{4\phi} B_{ca} {B^a}_d 
n^c n^d - G_2 n^a n^c B_{ac} \right) 
\ee
has indefinite sign. 

Following a similar procedure as in Sec.~\ref{subs:Case1} it is easy to 
see that:
\begin{prop}
Consider the {\em vacuum} viable Horndeski gravity with a null scalar 
field gradient $k_a := \nabla_a \phi$ and $\Box \phi = 0$. Then the 
effective stress-energy tensor \eqref{H-SEtensor} for the Horndeski 
scalar field $\phi$ is shearless. \end{prop}

\begin{proof}
The Raychaudhuri equation for null geodesic congruences \eqref{Raychaudhuri} reduces to
\be
2 \tilde{\sigma}^2 = \tilde{\sigma}_{ab}\tilde{\sigma}^{ab} = 
-R_{ab} k^a k^b
\ee
Contracting twice the Horndeski field equation \eqref{gabEOM} for {\em  
vacuum} viable Horndeski gravity with $k^a$, and with the assumption $\Box \phi = 
0$, one obtains
$$
R_{ab} k^a k^b = T_{ab} k^a k^b \,.
$$
Furthermore, 
\be
T_{ab} k^a k^b =
 \frac{1}{G_4} \left[ \left( -G_{3\phi}  +G_{4\phi\phi} 
+\frac{G_{2X}}{2} \right) k_a k_b 
 + \,G_{4\phi} B_{ab} + \frac{G_2}{2G_4} \, g_{ab} \right] k^a k^b =  0 \, ,
\ee
hence we can conclude that $ 2 \tilde{\sigma}^2 := 
\tilde{\sigma}_{ab}\tilde{\sigma}^{ab} = -R_{ab} k^a k^b = - T_{ab} k^a 
k^b = 0$. Since $\tilde{\sigma}^2$ is positive-definite, all the 
components of $\tilde{\sigma}_{ab}$ vanish identically, which completes 
the proof. \end{proof}

The stress-energy tensor of {\em vacuum} viable Horndeski gravity with a 
null scalar field gradient $k_a := \nabla_a \phi$ and $\Box 
\phi = 0$ at every spacetime point then reads
\be
T_{ab} = \rho k_a k_b + P h_{ab} +q_a k_b + q_b k_a \,,
\ee
which is not of type~III because $q^a$ is not necessarily spacelike nor 
orthogonal to the null vector $k^a$. However, if $G_2 \left( \phi, X=0 
\right)=0$ one has that the above stress-energy tensor reduces to
\be \label{H-typeIII}
T_{ab} = \rho k_a k_b + q_a k_b + q_b k_a \,, 
\ee
which instead belongs to the type~III family of the Hawking-Ellis 
classification.

In other words, we have shown that:
\begin{theorem}
Given the {\em vacuum} viable Horndeski gravity with a null 
scalar 
field gradient $k_a := \nabla_a \phi$, $\Box \phi = 0$, and 
$G_2 \left( \phi, X=0 \right)=0$, the  stress-energy tensor 
\eqref{H-SEtensor} for the Horndeski scalar field reduces to 
\eqref{H-typeIII}, which belongs to the type~III family of the 
Hawking-Ellis classification.
\end{theorem}

The stress-energy tensors (\ref{eq:typeIIIBD}) and 
(\ref{H-typeIII}) are not the most general Type~III stress-energy tensors 
because the null vector field $k^a$ originates from a gradient and is 
divergence-free, which are not general properties.

\section{Discussion and conclusions}
\label{sec:4}
\setcounter{equation}{0}

Assuming that the gravitational scalar field $\phi$ of first-generation 
scalar-tensor or Horndeski gravity is null everywhere, it is also geodesic 
and affinely parameterized, and the congruence of null geodesics with 
tangent $k_a=\nabla_a \phi$ is shear-free, non-twisting, and 
non-expanding. We have classified its effective stress-energy tensor, 
obtained by writing the field equations as effective Einstein equations.

In first-generation scalar-tensor gravity this effective $T_{ab}$ can be 
of only two types.  In the first case, it reduces to the well known null 
dust, or type~II in the Segr\'e-Pleba\'nski-Hawking-Ellis classification. 
In the second case, it contains an energy flux density $q^a$ which is 
spacelike, and we have a physical realization of Type~III stress-energy 
tensor (which can even become the simplified Type~III$_0$ 
of 
\cite{Martin-Moruno:2018coa, Martin-Moruno:2019kzc} in $\omega=0$ 
Brans-Dicke theory). Therefore, we have  an implementation of the 
ugly 
duckling type~III stress-energy tensor. This avatar of the ugly duckling 
is derived directly from the scalar-tensor Lagrangian and not from exotic 
Lagrangians constructed {\em ad hoc} (which seems to be the only avenue found as yet 
in Einstein gravity).  Explicit examples are not easy to find and are likely to 
be contrived. The most likely candidates for type~III effective 
stress-energy tensors in scalar-tensor gravity are Kundt spacetimes, 
however the known exact solutions of this kind in ``old'' scalar-tensor 
\cite{Tupper74, Ray:1977nr, Bressange:1997ey, Racsko:2018xcw, 
Siddhant:2020gkn} and in Horndeski \cite{Gurses:1978zs, Gurses:2016vwm, 
Gurses:2021ogc} theories have stress-energy tensors describing pure null 
dusts ({\em i.e.}, of type~II).

The situation in viable Horndeski gravity is more complicated, as an 
isotropic pressure appears unless $G_2 \left( \phi, X=0 \right)=0$ (in 
which case the discussion for first-generation scalar-tensor gravity 
applies again) and the energy flux density $q^a$ is  
neither orthogonal to $k^a$ nor spacelike.

We suggest a possible physical interpretation of the type~III 
energy-momentum tensor, with the obvious {\em caveat} that the null vector 
$n^a$ and, therefore, the 2-metric $h_{ab}$, density $\rho$, and vector 
$q^a$ are non-unique. The null dust part $\rho k_a k_b$ of the effective 
$T_{ab}$, with $k^a$ null and geodesic, describes coherent propagation of 
radiation, which is accompanied by a spacelike (therefore, non-causal) 
dissipation of energy in the direction transverse to $k^a$ and $n^a$, as 
in heat conduction. The interpretation of $q^a$ is essentially the same 
provided for the dissipative stress-energy tensor $T_{ab}= \rho u_a u_b +P 
h_{ab} +\pi_{ab}+q_a u_b$ when the four-velocity $u^a$ of a dissipative 
fluid is timelike and $q^a$ is spacelike \cite{Eckart:1940te}. It seems 
counterintuitive that the propagation of a beam at light speed would be 
compatible with the removal of energy from the beam, but energy is 
ill-defined for non-asymptotically flat geometries and the simpler 
$pp$-waves of general relativity exhibit energy-related features that are 
difficult to interpret \cite{Hayward:1993ph}, hence this objection may not 
be substantial after all.

The effective stress-energy tensor of the Brans-Dicke-like scalar 
field (in the so-called Jordan frame used in the present work) changes 
type  under a conformal transformation, which preserves 
 the lightlike nature of $\nabla^a \phi$.  A Type~III energy-momentum 
tensor in the Jordan frame of first-generation scalar-tensor gravity will 
become Type~II in the conformally transformed version, the Einstein frame, 
a property that makes it possible to generate and study Jordan frame 
geometries sourced by type~III stress-energy tensors using conformal 
mapping and known Einstein frame solutions associated with Type~II 
$T_{ab}$'s. The reader might have noticed the lack of exact solutions 
providing explicit examples of Type~III $T_{ab}$ in the previous section 
(which, of course, requires one to fix the coupling functions $G_i$). In a 
future publication we will search for such solutions using conformal 
mappings.
Incidentally, when the gradient $\nabla_a \phi$ of the gravitational 
scalar field of scalar-tensor (including Horndeski) gravity is lightlike, 
the effective ``fluid'' described by its stress-energy tensor does not \newpage
\noindent lead to a concept of effective ``temperature of gravity'' \footnote{This 
is because $q^a$ does not satisfy a known constitutive relation relating 
it to a temperature. For timelike $\nabla_a \phi = u_a$, $q_a$ satisfies 
the Eckart constitutive relation $q_a = -\mathcal{K} h_{ab} \left(\nabla^b 
\mathcal{T} + \mathcal{T} \dot{u}^b\right)$, which allows one to define an 
effective temperature of gravity.}, as it does when $\nabla_a \phi$ is 
timelike. Although not unexpected, this conclusion dashes the hope of 
extending the first-order thermodynamics of scalar-tensor and Horndeski 
gravity developed in \cite{Faraoni:2018qdr, Faraoni:2021lfc,
Faraoni:2021jri, Giusti:2021sku, Giardino:2022sdv, Giusti:2022tgq, 
Faraoni:2022gry, Miranda:2022wkz, Miranda:2022uyk, Faraoni:2022doe, 
Faraoni:2023hwu, Faraoni:2022jyd, Faraoni:2022fxo} to the null case.

We conclude that the ugly duckling of the 
Segr\'e-Pleba\'nski-Hawking-Ellis classification of stress-energy tensors 
may be a freak of nature of limited importance, but it is not physically 
impossible.
\section*{Acknowledgments} 
We are grateful to Marcello Miranda and Serena Giardino for useful 
discussions and to a referee for very helpful comments. N.~B. is 
grateful to Bishop's University for hospitality under the Tomlinson 
Visiting Scholars programme. This work is supported, in part, by the 
Natural Sciences \& Engineering Research Council of Canada through grant 
No.~2023-03234 (V.~F.) and an Undergraduate Student Research Award 
(R.~V.). The work of A.~G. has been carried out in the framework of the 
activities of the Italian National Group of Mathematical Physics [Gruppo 
Nazionale per la Fisica Matematica (GNFM), Istituto Nazionale di Alta 
Matematica (INdAM)].

\begin{appendices}

\section{$\nabla^a \phi$ as an eigenvector of the Ricci tensor and 
relation between Horndeski coupling functions}
\label{sec:AppendixA}
\renewcommand{\theequation}{A.\arabic{equation}}
\setcounter{equation}{0}

In viable Horndeski gravity, assuming $X\equiv 0$, $\Box \phi=0$, and 
$G_4=G_4(\phi)$, the gradient and the d'Alembertian of $X$ also vanish 
identically and Eq.~(\ref{whichlabeldoIput}) reduces to 
\be
G_{4\phi} R +G_{2\phi} -G_{3X} \nabla^c \phi \Box \nabla_c \phi 
+G_{3X}R_{ab} \nabla^a \phi \nabla^b \phi =0 \label{A1}
\ee
at $X=0$, $\Box\phi=0$. By applying the commutation relation
\be
\left( \nabla_a \nabla_b - \nabla_b \nabla_c \right) \omega_c 
={R_{abc}}^d  \omega_d
\ee 
to $\omega_a = \nabla_a \phi$ and using $\Box\phi=0$, one obtains
\begin{eqnarray}
\nabla^c \phi \Box \nabla_c \phi &=& 
\nabla^c \phi \nabla^a \nabla_a \nabla_c \phi = 
\nabla^c \phi \nabla^a \nabla_c \nabla_a \phi   \nonumber\\
&=& \nabla^c\phi \left( \nabla_c \Box \phi  +{R^a}_{cad} \nabla^d \phi 
\right) \nonumber\\
&=& \nabla^c\phi \left( \nabla_c \Box \phi + R_{cd}  \nabla^d \phi 
\right) \nonumber\\
&=&  R_{ab} \nabla^a \phi \nabla^b \phi  \label{hideki}
\end{eqnarray}
using the definition of the Ricci tensor $ R_{dc} \equiv {R_{dac}}^a$ and 
the symmetries of the Riemann tensor.

{\em In vacuo}, Eq.~(\ref{whichlabeldoIput}) 
reduces to 
\be
G_4 \left( R_{ab}-\frac{1}{2} \, g_{ab} R \right) -\nabla_a \nabla_b G_4 + 
\left( \Box G_4 -\frac{G_2}{2} -\frac{1}{2} \, \nabla_c\phi \nabla^c G_3 
\right) g_{ab} 
- \frac{G_{2X}}{2} \,\nabla_a \phi \nabla_b \phi +\nabla_{(a}\phi 
\nabla_{b)} G_3=0 \,,\label{A3}
\ee
taking the trace of which (and using $\Box G_4= G_{4\phi\phi} 
\nabla^a\phi\nabla_a\phi +G_{4\phi} \Box\phi=0$) produces  
\be
R=- \frac{2 G_2\left( \phi, X=0 \right)}{G_4(\phi)} \,.\label{A4}
\ee
Now, the contraction of Eq.~(\ref{A3}) with $\nabla^a\phi \nabla^b \phi$ 
yields
\be
R_{ab} \nabla^a\phi\nabla^b \phi =\frac{ G_{4\phi}}{G_4} \, 
\nabla^a\phi \nabla^b \phi \nabla_a\nabla_b \phi \,, \label{A5}
\ee
but
\be
\nabla^a\phi \nabla^b \phi \nabla_a\nabla_b \phi =k^a k^b B_{ab}=0 
\ee
due to the $k$-transversality of $B_{ab}$, hence
\be
R_{ab} \nabla^a\phi \nabla^b \phi=0 \,. \label{A6}
\ee
Using this result, it follows from Eq.~(\ref{hideki}) that the 
combination
\be
G_{3X} \left( - \nabla^c\phi   \Box \nabla_c \phi  +R_{ab} \nabla^a \phi 
\nabla^b \phi \right) 
\ee
appearing in Eq.~(\ref{A1}) vanishes when $\Box\phi=0$ and $X=0$, and 
Eq.~(\ref{A1}) then gives the relation 
between coupling functions
\be
G_{4\phi} R + G_{2\phi}=0 \label{mammamia}
\ee
at $\Box\phi=0, X=0$.
Now the comparison of Eqs.~(\ref{mammamia}) and (\ref{A4}) yields 
\be
\frac{ G_{2\phi}}{G_2} -\frac{2G_{4\phi}}{G_4}=0 
\ee
at $X=0, \, \Box\phi=0$. Furthermore, substituting Eq.~(\ref{A4}) into Eq.~(\ref{A3}) 
yields
\be
G_4 R_{ab} +\frac{G_2}{2} g_{ab} + \left( G_{3\phi} -G_{4\phi\phi} 
-\frac{G_{2X}}{2} \right) \nabla_a\phi \nabla_b \phi -G_{4\phi} 
\nabla_a\nabla_b \phi =0 
\ee
which, contracted with the gradient $\nabla^a\phi$, then gives
\be
R_{ab} \nabla^a\phi =-\frac{G_2}{2G_4} \, \nabla_b \phi = \frac{R}{4} \, 
\nabla_b \phi \,:
\ee
the scalar field gradient $\nabla^a\phi$ is an eigenvector of the Ricci  
tensor with eigenvalue $R/4$.
\end{appendices}

\end{document}